\documentclass[graybox]{svmult}

\usepackage{mathptmx}       
\usepackage{helvet}         
\usepackage{courier}        
\usepackage{type1cm}        
\usepackage{makeidx}         
\usepackage{graphicx}        
\usepackage{multicol}        
\usepackage[bottom]{footmisc}

\usepackage{amsmath }
\usepackage{amsfonts}
\usepackage{amssymb}

\newcommand{\be}{\begin{equation}}
\newcommand{\ee}{\end{equation}}
\newcommand{\sech}{\mathrm{sech}}

\makeindex             

\begin{document}


\title*{Localised nonlinear modes in the {\it PT}-symmetric
double-delta well  Gross-Pitaevskii equation
}

 \titlerunning{Localised modes in  \emph{PT}-symmetric Gross-Pitaevskii}
\author{I. V. Barashenkov and  D. A. Zezyulin}   
\institute{I. V. Barashenkov 
\at National Institute for Theoretical Physics, Western Cape, South Africa and  Department of Mathematics, University of Cape Town,  Rondebosch 7701. \email{Igor.Barashenkov@uct.ac.za}
\and  D.  A. Zezyulin \at Centro de F\'isica Te\'orica e Computacional and Departamento de F\'isica,
Faculdade de Ci\^encias da Universidade de Lisboa, Campo Grande, Edif\'icio C8, Lisboa  P-1749-016, Portugal. \email{dzezyulin@fc.ul.pt}      }
%
%
\maketitle

\abstract{We construct exact localised solutions of the {\it PT}-symmetric
Gross-Pitaevskii equation with an attractive cubic nonlinearity.
The trapping potential has the form of two $\delta$-function wells,
where one well loses particles while the other one is fed with atoms at an equal rate.
The parameters of the constructed solutions are expressible
 in terms of the roots of a system of two
transcendental algebraic equations.
We also furnish a simple analytical treatment of the linear
Schr\"odinger equation with the \emph{PT}-symmetric double-$\delta$ potential.
}

\section{Introduction}
\label{Intro}

We consider the Gross-Pitaevskii equation,
\be
i  \, \Psi_t + \Psi_{xx} - V(x)
 \Psi + g  |\Psi|^2  \Psi= 0,
\label{A0}
\ee
with $g \geq 0$ and the  \textit{PT}-symmetric potential
\be
\label{A00}
V(x)= U(x) + iW(x), \qquad U(-x)=U(x),   \quad  W(-x)=-W(x).
\ee

 The system
 (\ref{A0})--(\ref{A00})  was employed to model the dynamics 
 of the self-gravitating 
  boson condensate trapped in a confining potential   $U(x)$. The imaginary coefficient 
    $iW$ 
    accounts for  the particle leakage   --- in the region where $W(x)<0$ --- and
     the compensatory injection of atoms in the region where $W(x)>0$ \cite{Klaiman,Cartarius1}. 
     
     The  same equation was used to  describe  the stationary light beam propagation in the self-focusing Kerr medium. 
     In the optical  context, $t$ stands for  the  scaled propagation distance while 
       $\Psi$ is  the complex electric-field envelope.
         The real part of the potential ($U$)  is associated with the refractive index guiding, while  the imaginary part ($W$) 
         gives the optical gain and loss distribution  \cite{Musslimani}.

We are interested in localised solutions of this equation, that is, solutions with the asymptotic behaviour
$\Psi(x,t) \to 0$ as $x \to \pm \infty$. We also require  that
\be
\int_{-\infty}^\infty |\Psi|^2 dx =1.
\label{norm0}
\ee
In the context of leaking condensate with injection, the normalisation condition 
\eqref{norm0} implies that 
the total   number  of particles  in the condensate is kept at a constant level.

In this study, we consider stationary solutions of the form $\Psi(x,t)=\psi(x)e^{i \kappa^2 t}$, where
$\kappa^2$ is real and the spatial part of the eigenfunction
  obeys
\be
-\psi_{xx} +V(x)
 \psi - g\psi|\psi|^2 = -\kappa^2 \psi.
\label{A1}
\ee
Assuming that the potential satisfies $V(x) \to  0$ as $x \to \pm \infty$, $\kappa^2$ has to be taken positive.
For definiteness, we choose the real  $\kappa$  to be  positive as well.
The equation \eqref{A1} will be solved under 
 the normalisation constraint
\be
\int_{-\infty}^\infty |\psi|^2 dx =1,
\label{norm}
\ee
stemming from the condition \eqref{norm0}.


Our study will be confined to the \emph{PT}-symmetric  solutions
of equation \eqref{A1}, that is, solutions satisfying
\be
\psi(-x)= \psi^*(x).
\label{pspt}
\ee
Typically, 
stationary solutions supported by  \emph{PT}-symmetric potentials can be brought to the  form \eqref{pspt}
by a suitable constant phase shift.

With an eye to the  forthcoming study of the  jamming anomaly \cite{jam},
we consider  a \textit{PT}-symmetric  potential of the special
form:
\be
V(x)= -(1- i \gamma) \delta (x+L/2)
- (1+ i \gamma) \delta (x-L/2).
\label{C90}
\ee
Here $\gamma \geq 0$ and $L>0$. 
 The $V(x)$ is an idealised   potential  consisting of two infinitely deep
wells. The right-hand well is leaking particles, 
while its left-hand counterpart is injected with atoms at an equal  constant rate $\gamma$.

Previous analyses
of the double-delta cubic Gross-Pitaevskii equation
 focussed mainly on the situation with no gain or loss ---
that is, on the potential \eqref{C90} with $\gamma=0$.
Using a combination of analytical and numerical tools,
Gaididei, Mingaleev and Christiansen \cite{GMC}
demonstrated the spontaneous breaking of the left-right symmetry by
localised solutions.
Subsequently,
Jackson and Weinstein \cite{JW}
 performed geometric analysis of the symmetry breaking
 and classified the underlying
 bifurcations  of stationary solutions. Besides the absence of gain and loss,
 the mathematical setting of Ref.\cite{JW} was different from our present problem in the lack of
 the normalisation condition \eqref{norm0}.

Studies of the  \textit{PT}-symmetric model with $\gamma \ne 0$ were pioneered by Znojil and Jakubsk\'y
 who analysed the linear Schr\"odinger equation with  point-like gain and loss (but no wells) 
on a finite interval \cite{Znojil1,Znojil2}. 
 The double-well  potential  \eqref{C90} was proposed by Uncu and Demiralp \cite{Uncu}
  whose paper also focussed on the linear equation --- yet on the infinite line. 
 Cartarius, Haag, Dast, and Wunner \cite{Cartarius1,Cartarius2}
 considered both linear and nonlinear Gross-Pitaevskii model. The numerical study of Refs.\cite{Cartarius1,Cartarius2}
 identified a branch
of localised nonlinear modes bifurcating from eigenvalues of the linear operator in \eqref{A1}.  

In this contribution, we get an analytical handle on the \emph{PT}-symmetric double-$\delta$ 
problem, linear and, most importantly, nonlinear. In the linear situation
(equation \eqref{A1} with $g=0$) we provide a 
mathematical interpretation and verification of the numerics reported in \cite{Cartarius1,Cartarius2}.
In the nonlinear case ($g \neq 0$), the analytical consideration allows us to
advance beyond the numerical conclusions of the previous authors.
In particular, we demonstrate the existence of infinite families of localised solutions 
with multiple humps and dips between the two potential singularities.


\section{Linear Schr\"odinger equation with complex double-$\delta$ well potential}

Relegating the analysis of the full nonlinear equation    \eqref{A1}, \eqref{C90}
to the
subsequent sections,
here we consider its  linear
particular case ($g=0$).
The normalised  eigenfunction pertaining to the eigenvalue $-\kappa^2$
is given by
\be
\psi(x)= \left\{
\begin{array}{lr}
   \frac{e^{i \phi} + e^{ \kappa L- i \phi}}{2 \sqrt N}  e^{\kappa x}, & x \leq -L/2;  \\ \\
   \frac{ \cosh (\kappa x+ i \phi)}{\sqrt{N}}, & -L/2 \leq x \leq L/2;  \\ \\
 \frac{e^{-i \phi} + e^{ \kappa L+ i \phi}} {2 \sqrt N}   e^{-\kappa x}, & x \geq L/2.
 \end{array}
 \right.
 \nonumber
 \ee
Here $\kappa$ is a positive root of the transcendental equation    \cite{Uncu,Cartarius1,Cartarius2}
\be
e^{- 2 \kappa L}= \frac{\gamma^2+ (2 \kappa-1)^2}{\gamma^2+1},
\label{C6}
\ee
while $\phi$ and $N$ are readily expressible through $\kappa$.
The secular equation (\ref{C6}) was solved numerically in \cite{Cartarius1,Cartarius2}.
Here, we analyse it
without resorting to the help of computer.

To this end, we express
 $\gamma$ as an explicit function of $\kappa$:
\be
\gamma^2= \frac{4 \kappa (1-\kappa)}{1- e^{-2 \kappa L}}-1.
\label{C7}
\ee
Instead of evaluating eigenvalues $\kappa$ as the parameter $\gamma>0$ is varied,
we identify the range of positive $\kappa$ where the function $\gamma^2(\kappa)$ is positive.
We prove the following

\begin{proposition}
Regardless of the value of $L$, there is a finite interval of $\kappa$ where $\gamma^2>0$.
When $L<2$, the interval is $0< \kappa< \kappa^{(b)}$, and when
$L>2$, the interval is $\kappa^{(a)}< \kappa<\kappa^{(b)}$. Here $\kappa^{(a)}$
and $\kappa^{(b)}$ are  dependent on $L$, with $0< \kappa^{(a)}<\kappa^{(b)} <1$.
\end{proposition}

\begin{proof}
The inequality $\gamma^2>0$ amounts to $k_1< \kappa <k_2$, where
the endpoints of the interval $(k_1,k_2)$ also depend on $\kappa$:
\[
k_1(\kappa) =\frac{1-e^{-\kappa L}}{2}, \quad k_2(\kappa) = \frac{1+ e^{-\kappa L}}{2}.
\]
If $L<2$,
the quantity $k_1(\kappa)$
is smaller than $\kappa$
for all $\kappa>0$. If, on the other hand, $L>2$, the graph of the function
$y=k_1(\kappa)$ lies above
 $y=\kappa$ in the interval $0 \leq \kappa < \kappa^{(a)}$
and below $y=\kappa$  in the interval $ \kappa^{(a)}< \kappa< \infty$.
Here $\kappa^{(a)}=\kappa^{(a)}(L)$ is the (unique) root of the equation
$k_1(\kappa)=\kappa$.

On the other hand,  the function $k_2(\kappa)$
 is greater than $\kappa$
when $0<\kappa<\kappa^{(b)}$ and smaller than $\kappa$ when $\kappa>\kappa^{(b)}$.
Here $\kappa^{(b)}=\kappa^{(b)}(L)$ is the  root of the equation $k_2(\kappa)=\kappa$.
(There is a unique root for all $L>0$.)

Note that in the range of the $L$ values where the root $\kappa^{(a)}$ exists
--- that is, in the region $L>2$ ---  we have $\kappa^{(a)} < \kappa^{(b)}$.
Since $\gamma^2(1)<0$, we have $\kappa^{(b)}<1$.
 $\Box$
\end{proof}

Our next result concerns the number of eigenvalues arising for various $\gamma$.
Again, instead of counting branches of the function $\kappa(\gamma)$, we
identify regions of monotonicity of the inverse function, $\gamma(\kappa)$.
These are separated by the points of local extrema (stationary points).

\begin{proposition}
When $L<1$, the function $\gamma(\kappa)$ is monotonically decreasing as
$\kappa$ changes from 0 to $\kappa^{(b)}$, with $\kappa^{(b)}$ defined above.
When $L>1$, the function $\gamma(\kappa)$ has a single local maximum
at $\kappa=\kappa_c$ (where $\kappa_c< \kappa^{(b)}$).
\end{proposition}

\begin{proof}
Stationary points of the function $\gamma^2(\kappa)$ are given by zeros of
\be
\frac{d \gamma^2}{d \kappa}=
\frac{(1-2\kappa)}
{2\kappa L \sinh^2(L\kappa)}
\left[ f(\kappa)-g(\kappa) \right],
\label{de}
\ee
where
\[
f= \frac{e^{2L \kappa} -1}{2L \kappa},    \quad
g= \frac{1-\kappa}{1-2\kappa}.
\]
We consider \eqref{de} for $0 < \kappa<1$.
(Note that $\kappa=1/2$ is not a zero of $d \gamma^2/ d\kappa$.)
When  $1/2<\kappa<1$, the function $g$ is strictly negative; hence
stationary points may only lie in the interval $0< \kappa < 1/2$.


%

\begin{figure}[t]
\center
\includegraphics[width=8cm]{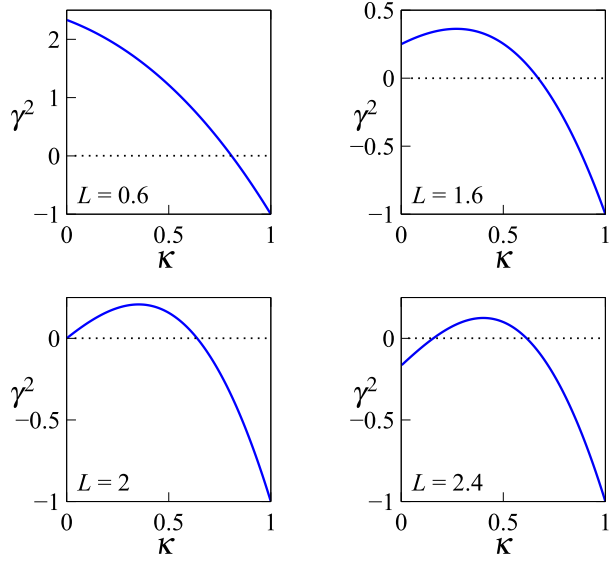}
\caption{The function (\ref{C7}) plotted for several representative values of $L$.
In the  interval of $\kappa$ where $\gamma^2 \geq 0$, the function gives the square of the gain-loss coefficient.
 The inverse function $\kappa(\gamma^2)$ is obtained simply by flipping the part of the graph with $\gamma^2 \geq 0$
  about the $\gamma^2 =\kappa$ line
(see Fig.\ref{Fig2}).
}
\label{Fig1}       
\end{figure}

Assume, first, that $L<1$
and expand  $f(\kappa)$ and $g(\kappa)$ in powers of $\kappa$:
\begin{align}
f= 1+ \sum_{n=1}^\infty  f_n \kappa^n, \quad  f_n= \frac{(2L)^n}{(n+1) !} \label{sf}
\\
g= 1+ \sum_{n=1}^\infty    g_n \kappa^n,  \quad g_n= 2^{n-1}. \label{sg}
\end{align}
The series \eqref{sf} converges in the entire complex plane of $\kappa$ while the 
series \eqref{sg}  converges in the disc $|\kappa|< 1/2$.
Noting that
$f_n<g_n$ for all $n$, we conclude that
 $f(\kappa)<g(\kappa)$ for all $0< \kappa< 1/2$.
 This implies that for any $L<1$, the function $\gamma^2(\kappa)$ decreases monotonically 
 as $\kappa$ changes from 0 to 1.

Let now $L>1$. 
The values of $f$ and $g$ at the origin are equal
while their slopes are not:
\[
\left. \frac{df}{d \kappa} \right|_{\kappa=0}= L,
\quad
\left.  \frac{dg}{d \kappa} \right|_{\kappa=0}=1.
\]
Consequently, the graph  of $f(\kappa)$ lies above the graph of $g(\kappa)$ 
as long as  $\kappa$ remains sufficiently
close to the origin.
At the opposite end of the interval, that is, in the vicinity of $\kappa=1/2$,
 the graph of $g(\kappa)$ lies above $f(\kappa)$.
Therefore the equation $f(\kappa)=g(\kappa)$  has (at least one) root $\kappa_c$ in the interval $0< \kappa< 1/2$.
This root   emerges from the point $\kappa=0$
as soon as $L$ becomes greater than  $1$.

\begin{figure}[t]
\center
\includegraphics[width=8.5cm]{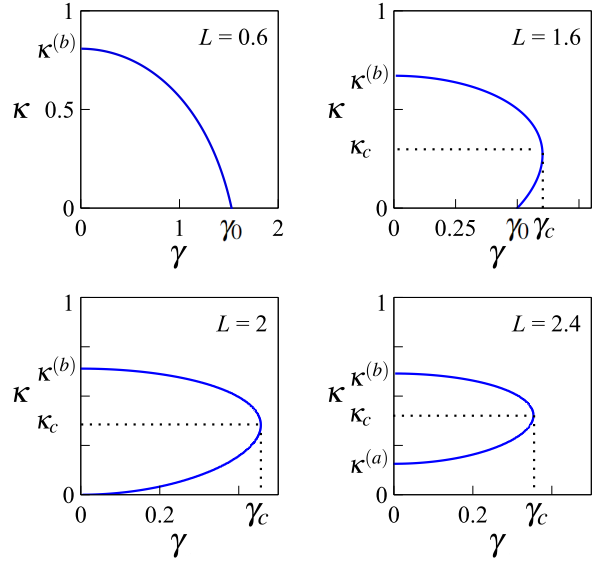}
%
%
\caption{Positive roots of \eqref{C6}  vs the gain-loss coefficient $\gamma$, for several representative values of $L$.
}
\label{Fig2}       
\end{figure}

%

 To show that no additional stationary points can emerge as $L$ is further increased, assume
 the contrary --- assume
 that a pair of stationary points is born
 as $L$ passes through a critical value $L_*$
 (where $L_*> 1$). At  the bifurcation value $L=L_*$,
 the newborn stationary points are equal; we denote them $\kappa_*$.
 When $L=L_*$, the equality
 \be
  \left. \frac{df}{d \kappa} \right|_{\kappa=\kappa_*}   = \left. \frac{d g}{d \kappa} \right|_{\kappa=\kappa_*}
  \label{str2}
  \ee
  should be fulfilled along
 with 
 \be
 f(\kappa_*)=g(\kappa_*).
 \label{fgk}
 \ee
 Solving the system \eqref{str2}, \eqref{fgk}  yields
  $L_*=(1-2 \kappa_*)^{-1}$. This can be written as
  \be
  L_*=1+q, \label{Lx}
  \ee
  where  we have defined $q= 2 \kappa_* L_*$. Making use of \eqref{Lx},  equation \eqref{fgk} becomes
  \[
  \frac{e^q-1}{q} = 1+ \frac{q}{2}.
  \]
  Expanding the left-hand side in powers of $q$ one can readily check that it is greater than the right-hand side for 
  any $q>0$; hence the above equality can  never be satisfied.
  This  proves that the stationary point $\kappa_c$ of the function $\gamma^2(\kappa)$ is single.
  
   Since $d \gamma^2/ d \kappa|_{\kappa=0}>0$,
  the  stationary point $\kappa_c$
  is a maximum.  $\Box$
  \end{proof}

The above propositions are illustrated by
Fig.\ref{Fig1} which shows $\gamma^2(\kappa)$ with $L<1$ (a);
$1<L<2$ (b); $L=2$ (c), and $L>2$ (d).

Our conclusions are sufficient to determine the shape of
the inverse function, $\kappa(\gamma)$.
When $L<1$, there is a single positive branch of $\kappa(\gamma)$
which decays, monotonically, as $\gamma$ is increased  from zero to $\gamma_0$
(Fig. \ref{Fig2}(a)).
As $\gamma$ reaches $\gamma_0$, the quantity $\kappa$ drops to zero and 
the  eigenvalue $-\kappa^2$ collides with the continuous spectrum.
Since $\gamma^2(0)= 2/L-1$, we can obtain the critical value of $\gamma$ exactly:
 $\gamma_0= \sqrt{2/L-1}$.

When $L$ is taken between $1$ and $2$, the function $\kappa(\gamma)$ has two branches
(Fig. \ref{Fig2}(b)).
Along the
 monotonically decreasing branch, $\kappa$ drops from $\kappa^{(b)}$ to $\kappa_c$
 as $\gamma$ is raised from 0 to $\gamma_c$.
 In addition, there is a monotonically increasing branch with $\gamma_0< \gamma< \gamma_c$.
 Here, $\kappa$ grows from 0 to $\kappa_c$ as
$\gamma$ is increased  from $\gamma_0$ to $\gamma_c$.
The two eigenvalues  merge and become complex as  $\gamma$ is raised through
 $\gamma_c$.

Finally, when $L \geq 2$, the monotonically decreasing and increasing branch of
 $\kappa(\gamma)$ exist over the same interval $0< \gamma< \gamma_c$
 (Fig. \ref{Fig2} (c,d)).
As $\gamma$ grows from 0 to $\gamma_c$, one branch of $\kappa$
grows from $\kappa^{(a)}$ to $\kappa_c$ whereas the
other one decreases from $\kappa^{(b)}$ to $\kappa_c$.

These conclusions are in agreement with the numerical results of \cite{Cartarius1,Cartarius2}.

\section{\emph{PT}-symmetric Gross-Pitaevskii with   variable-depth wells }

Proceeding to the nonlinear situation ($g \neq 0$),
  it is convenient to transform the stationary equation (\ref{A1})   to
\be
\varphi_{\tau \tau } + \lambda \left[\delta(\tau+T)+ \delta(\tau-T)\right] \varphi
- i \eta [\delta(\tau+T)- \delta(\tau-T)] \varphi
  + 2  \varphi|\varphi|^2 = \varphi,
\label{A100}
\ee
with
\be
\tau= \kappa x, \quad T=\kappa \frac{L}{2}, \quad
\varphi=  \sqrt{\frac{g}{2}} \frac{\psi}{\kappa}, \quad  \lambda= \frac{1}{\kappa},
\quad  \eta=\frac{\gamma}{\kappa}.
\nonumber
\ee
Here $\eta \geq 0$, $\lambda>0$, and $T>0$.
In equation (\ref{A100}) the chemical potential has been normalised to unity at
the expense of making the well depths, $\lambda$, variable.
The normalisation constraint (\ref{norm})
acquires the form
\be
\int_{-\infty}^\infty |\varphi|^2 d \tau = \frac{\lambda}{2}g,
\label{norm2}
\ee
while the symmetry condition (\ref{pspt}) translates into
\be
\varphi^*(\tau) = \varphi(-\tau).
\label{PT}
\ee

Consider a solution $\varphi(\tau)$  of the equation \eqref{A100} and denote
 $N= \int |\varphi|^2 d \tau$ the corresponding number of particles. The number of particles is a function of $\lambda$, $\eta$ and $T$: $N=N(\lambda, \eta, T)$.
Setting $g$ to a particular value, the constraint (\ref{norm2}) defines a two-dimensional surface in the $(\lambda, \eta, T)$ space:
\[
\frac{1}{\lambda} N(\lambda, \eta, T)= \frac{g}{2}.
\]
For any fixed $L$,
the ``nonlinear eigenvalue"
$\kappa=\lambda^{-1}$ becomes an (implicit) function of $\gamma$:
\[
\kappa N  \left(\frac{1}{\kappa}, \frac{\gamma}{\kappa}, \frac{\kappa L}{2}
\right)= \frac{g}{2}.
\]
The purpose of our study is to construct the solution $\varphi(\tau)$
and determine
this function.

It is fitting to note here that
the equation (\ref{A1}) with $g = 0$ can also be transformed to the form (\ref{A100}) --- where 
one just needs to drop the cubic term.
In this  case, the number of particles  is not fixed though; that is, the constraint (\ref{norm2}) does not need to be satisfied.
The relation between  $\kappa$ and $\gamma$ --- equation \eqref{C6} --- arises
as a secular equation for an eigenvalue problem.

\section{Particle moving in a mexican-hat potential}
\label{Mexican}

In the external region $|\tau | \geq T$, the
\emph{PT}-symmetric solutions with the boundary conditions $\varphi(\pm \infty)=0$
have the form
\begin{eqnarray}
\varphi(\tau)= e^{-i \chi} \sech (\tau+\mu), \quad & \tau \leq -T;     \nonumber    \\
\varphi(\tau) = e^{i \chi} \sech (\tau -\mu), \quad &  \tau \geq T.     \label{A6}
\end{eqnarray}
Here $\mu$ is an arbitrary real value, positive or negative, determining the amplitude of $\varphi$,
and $\chi$ is an arbitrary real phase.
The solution $\varphi(\tau)$ in the internal region $ |\tau| \leq T$ will be matched to the values of \eqref{A6}
at the endpoints
 $\tau = \pm T$:
\be
\varphi( \pm T)= e^{  \pm  i \chi} \sech (\mu-T).   \label{B3}
\ee
Integrating \eqref{A100} across the singularities and using \eqref{A6} once again, we obtain the
matching conditions for the derivatives as well:
\begin{eqnarray}
\left.  {\dot  \varphi}   \right|_{\tau = T-0}  = e^{i \chi}  \sech(\mu-T)   [  i \eta
+\lambda   +  \tanh(\mu-T)], \nonumber   \\
\left.  {\dot \varphi}   \right|_{\tau = -T+0}= e^{ -i \chi}  \sech(\mu-T)   [  i \eta
-\lambda    -  \tanh(\mu-T)].
\label{B2}
\end{eqnarray}
Here the overdot stands for the differentiation with respect to $\tau$.

To construct the solution between $-T$ and $T$ it is convenient to
interpret the modulus and the phase of $\varphi$ as polar coordinates of a particle
on the plane,
\[
\varphi= r e^{i \theta},
\]
and the coordinate $\tau$ as time.
The newtonian particle moves in the radial-symmetric mexican hat-shaped potential $U(r)=\frac12 (r^4-r^2)$. Hence the angular momentum
\be
\ell = {\dot \theta} r^2
\label{B5}
\ee
and the energy
\be
\mathcal E= \frac{{\dot r}^2}{2} + U_{\rm eff} (r), \quad U_{\rm eff}= \frac{\ell^2}{2 r^2} + \frac{r^4-r^2}{2}
\label{B1}
\ee
are conserved. The effective potential for the radial motion  is shown in Fig.\ref{Fig_Ueff}(a).


\begin{figure}[t]
\center
\includegraphics[width=8cm]{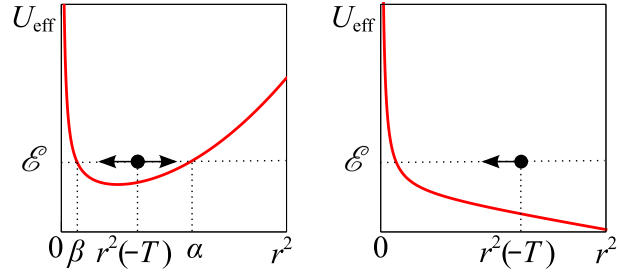}
%
%
\caption{The left panel shows $U_{\rm eff}(r)$, the effective potential of radial motion
defined in (\ref{B1}).
Two arrows indicate two possible directions of motion of the fictitious particle.
Whether the  particle starts with a positive or negative radial velocity,
it will run into a turning point.
Dropping the quartic term from (\ref{B1}) gives the  effective potential for the
linear equation  (right panel).  This time the turning point will only be run into if ${\dot r}(-T)<0$.  
}
\label{Fig_Ueff}       
\end{figure}

The particle starts its motion at time $\tau=-T$ and ends at $\tau=T$.
The boundary conditions follow from
(\ref{B3})-(\ref{B2}):
\begin{eqnarray}
r(\pm T)= \sech(\mu-T),   \label{B11} \\
{\dot r}(-T)= -{\dot r}(T)=  - (\lambda+\xi)  \sech(\mu-T),    \label{B12} \\
\theta(\pm T)= \pm \chi,      \label{B120}   \\
{\dot \theta}(\pm T)= \eta,
\label{B4}
\end{eqnarray}
where
\[
\xi= \tanh (\mu-T).
\]
The parameter $\xi$ satisfying $\lambda+\xi>0$ requires a negative initial velocity,
 ${\dot r}(-T)<0$, and   $\lambda+\xi<0$  corresponds to   an outward initial motion:
  ${\dot r}(-T)>0$.

Equations (\ref{B5}) and (\ref{B4}) imply that the conserved angular momentum has
a positive value:
\be
\ell = \eta \, \sech^2 (\mu-T)
\label{A9}
\ee
and so $\theta$ grows as $\tau$ varies from $-T$ to $T$. (Hence  $\chi>0$.)
The energy of the particle is found by substituting (\ref{B11})-(\ref{B12}) in (\ref{B1}):
\be
\mathcal E= \frac12   (1-\xi^2)              \left[  (\lambda  + \xi)^2 -\xi^2   +\eta^2\right].
\label{A33}
\ee

Since the energy \eqref{B1} includes the square of ${\dot r}$ but not ${\dot r}$ itself,
the information about the initial direction of motion becomes lost in the expression \eqref{A33}.
In fact, by using the value of energy instead of the boundary condition \eqref{B12} 
we are acquiring spurious solutions.  These solutions have the wrong sign of ${\dot r}(\tau)$ as $\tau \to -T+0$
and do not satisfy \eqref{B12}. 
Fortunately  we remember that the sign of $\left. {\dot r} \right|_{\tau \to -T+0}$ should be  opposite to that of $\lambda+\xi$.
This simple rule will be used to filter out  the spurious roots  in section \ref{Transcendental}.

The radial trajectory $r(\tau)$ for  a
\emph{PT}-symmetric solution satisfying \eqref{PT} 
 should be given by an even function and
the trajectory should
have a turning point at $\tau=0$: ${\dot r}(0)=0$.
The separable equation (\ref{B1}) has two
even solutions,
\be
r_A^2(\tau) = (\alpha-\beta) \mathrm{cn}^2  \left( \sqrt{2\alpha+ \beta - 1}\tau, k \right)+ \beta
\label{A15}
\ee
and
\be
r_{B}^2(\tau) =  (\alpha-\beta) \mathrm{cn}^2  \left( K-\sqrt{2 \alpha+\beta-1 }\tau, k \right)+ \beta.
\label{A16}
\ee
The Jacobi-function solutions (\ref{A15}) and (\ref{A16})  are parametrised by two parameters, $\alpha$ and $\beta$,
where $\alpha \geq \beta \geq 0$ and $\alpha+\beta >1$.
These are related to $\ell$ and $\mathcal E$ via
\begin{eqnarray}
\ell^2= \alpha \beta (\alpha+\beta-1),  \label{A31} \\
2   \mathcal E= (\alpha+\beta)(\alpha+\beta-1) -\alpha \beta.  \label{A32}
\end{eqnarray}
The
elliptic modulus $k$ is given by
\[
k^2= \frac{\alpha-\beta}{2 \alpha +\beta-1},
\]
and $K(k)$ in (\ref{A16}) is the complete elliptic integral of the first kind.
Eliminating $\ell$ and $\mathcal E$ between (\ref{A9}) and (\ref{A31}),
(\ref{A33}) and (\ref{A32}) we get
\begin{eqnarray}
\eta^2(1-\xi^2)^2= \alpha \beta (\alpha+\beta-1),  \label{A34} \\
(1-\xi^2) (\lambda^2+\eta^2+ 2 \lambda \xi)= (\alpha+\beta)(\alpha+\beta-1) -\alpha \beta.   \label{A35}
\end{eqnarray}

%

\begin{figure}[t]
\center
\includegraphics[width=8cm]{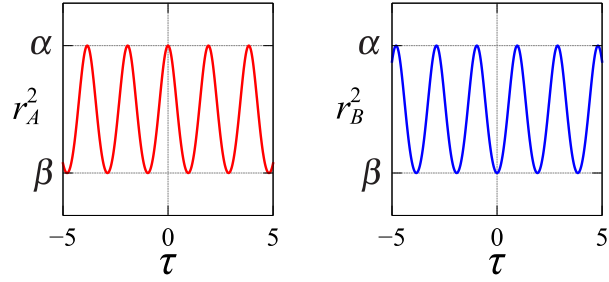}
%
%
\caption{Two even solutions of equation \eqref{B1}. (a): the maximum-centred solution, Eq.~(\ref{A15}).
(b): the minimum-centred solution,  Eq.~(\ref{A16}). In both panels $\alpha=2$ and $\beta=0.5$.
}
\label{rAB}     
\end{figure}

The solution $r_A$ is maximum-centred and $r_B$ is minimum-centred (see Fig.~\ref{rAB}).

For the purposes of our study, it is convenient to have a relation between $\alpha$, $\beta$, and $\xi$
not involving the gain-loss parameter.
Eliminating $\eta^2$ between (\ref{A34}) and (\ref{A35})
 we obtain
\be
(\lambda+\xi)^2=S^2,
\label{D200}
\ee
 where
\be
S^2=
 \frac{  (\alpha+\beta) (\alpha +\beta-1)-
 \alpha \beta }{1-\xi^2}
  - \frac{\alpha \beta (\alpha+\beta   -1)}{(1-\xi^2)^2} +\xi^2.
 \label{S}
\ee
The structural relation \eqref{D200}-\eqref{S} will prove useful in what follows.

\section{Boundary conditions and normalisation constraint}

The boundary conditions (\ref{B11}) give a transcendental equation
\begin{subequations}
\be
\beta+ (\alpha-\beta) \mathrm{cn}^2\left(y ,k \right)= 1-\xi^2,
 \label{A37}
\ee
for the $r_A$ and
\be
\beta+ (\alpha-\beta) \mathrm{cn}^2\left( K-y,k \right)= 1-\xi^2
 \label{A36}
\ee
\end{subequations}
for the $r_B$ solution. Here
\[
y= \sqrt{2 \alpha+\beta-1} T.
\]
Note that there is a simple correspondence between
equations \eqref{A37} and \eqref{A36}. Namely, if we assume that $\alpha, \beta$ and $\xi$
in \eqref{A37} and \eqref{A36} are given,
and denote $\widetilde{y}$ the value of $y$ satisfying \eqref{A36}, then
 $y=K-\widetilde{y}$ will satisfy \eqref{A37}.

The linear ($g=0$) Schr\"odinger equation \eqref{A1}, \eqref{C90}  corresponds to
the newtonian particle moving in the effective potential without the quartic barrier at
large $r$.
In this case the particle  can only run into a
turning point  if ${\dot r}(-T)<0$.
(See Fig.\ref{Fig_Ueff}(b)). On the other hand,
when  the quartic barrier is present,
the  particle will stop no matter whether ${\dot r}(-T)$ is negative or positive
(Fig.\ref{Fig_Ueff}(a)).

Consider, first, the solution $r_A(\tau)$ and
assume that  ${\dot r}(-T)>0$. 
The simplest trajectory satisfying the boundary conditions \eqref{B11}-\eqref{B12}
 describes
the particle starting with a positive radial velocity at $\tau=-T$,
reaching the maximum $r^2=\alpha$ at $\tau=0$ and returning to the starting point at $\tau=T$.
We use $T_0$ to denote the corresponding return time, $T$.
Coexisting with this solution are  longer  trajectories that reach the maximum $r^2=\alpha$ not once but $2n+1$ times
($n \geq 1$), namely,
at $\tau= 2 m \Theta$, where
\be
\Theta =  \frac{K(k)}{\sqrt{2 \alpha+\beta-1}}
\label{hp}
\ee
is the half-period of the  function $\mathrm{cn}^2(\sqrt{2 \alpha+\beta-1}\tau,k)$ and $m=0, \pm 1, \pm 2, ..., \pm n$.
These trajectories have the same values of $\alpha$ and $\beta$
(the same apogee and perigee)  but different
return times, $T= T_0+ 2n \Theta$.
Since the trajectory reaches its apogee $2n+1$ times and pays $2n$ visits to its minimum value of
$r^2=\beta$, we are referring to these solutions  as $(2n+1)$-hump, $2n$-dip nonlinear modes.

In contrast to these, the $r_A$ solution with ${\dot r}(-T)<0$
will be visiting its minimum $2n$ times ($n \geq 1$),  at $\tau= (1-2n)\Theta, ... (2n-1) \Theta$,
but will only pay $2n-1$ visits to its maximum. These trajectories will be
classified as $2n$-dip,
 $(2n-1)$-hump  modes.  The corresponding return times are $T=   2n \Theta- T_0$.

Turning to the minimum-centred solutions, we consider trajectories with
${\dot r}(-T)<0$ first.  
The simplest  $r_B$ solution describes  the particle starting with a negative velocity at $\tau=-T$,
reaching its perigee $r^2=\beta$ at $\tau=0$ and returning to the starting point at $\tau=T$, where $T=\Theta-T_0$
and $T_0$ was introduced above.
The $r_B$ solutions with more bounces visit the minimum $r$ not once but $2n+1$ times
($n \geq 1$),
at $\tau=2m \Theta$, $m=0, \pm 1, ..., \pm n$. The corresponding
return time is $T=  (2n+1) \Theta-T_0$. With their $2n+1$ local minima and $2n$ maxima,
these solutions are referred to as the $(2n+1)$-dip, $2n$-hump nonlinear modes.

Finally,
the $r_B$ solution with ${\dot r}(-T)>0$ reaches its apogee
$2n$ times
($n \geq 1$), that is, at $\tau= (2m+1) \Theta$, with $m=-n-1, ..., n$.
The return time is $T=T_0+(2n-1) \Theta$.
The trajectory pays $2n-1$ visits to its minimum value of $r^2=\beta$;
hence we classify this solution as
the $2n$ hump, $(2n-1)$-dip modes.


For the fixed $\alpha$ and $\beta$, the return time  $T_0$ is given by  the smallest positive root of \eqref{A37}. (Note that $T_0< \Theta$.)  Other roots
of this equation
 are
$T_0+2\Theta$, $T_0+4\Theta$, ...,
and $2\Theta-T_0$, $4\Theta-T_0$, ... .
On the other hand, the smallest positive root of \eqref{A36} is $\Theta-T_0$,
with other roots being $3 \Theta-T_0$, $5 \Theta-T_0$, ..., and $T_0+ \Theta$, $T_0+ 3 \Theta$, ... .

The normalised return time
\be
\frac{T}{\Theta} = \frac{y}{K(k)}
\label{RT}
\ee
 provides a simple tool for the identification of the nonlinear mode.
Indeed,   an $r_A$ solution  with $T/\Theta$
between  $2n$ and $2n+1$  has $2n+1$ humps, $2n$ dips and ${\dot r}(-T)>0$.
On the other hand, an $r_A$ solution  with $T/\Theta$
between $2n-1$ and $2n$  has  $2n$ dips, $2n-1$ humps  and ${\dot r}(-T)<0$.
Similarly, an $r_B$ solution will have $2n$ humps, $2n-1$ dips, and ${\dot r}(-T)>0$ --- if
 $T/\Theta$ lies between $2n-1$ and $2n$,
or  $2n+1$ dips, $2n$ humps, and ${\dot r}(-T)<0$ --- if $T/\Theta$ is
 between  $2n$  and  $2n+1$.

Evaluating the number of particles 
\[
N= \int_{-\infty}^{-T} \sech^2(\tau+\mu) d \tau+ \int_{-T}^T r_A^2 d  \tau + \int_T^\infty \sech^2(\tau-\mu) d \tau
\]
 and substituting in the normalisation constraint (\ref{norm2}), the constraint is transformed into
 \be
  \zeta_A(\alpha, \beta, y, \lambda)=\xi,
 \label{xia}
 \ee
  where the function $\zeta_A$ is defined by
\be
\zeta_A=  \frac{\alpha+\beta-1}{\sqrt{2\alpha+\beta-1}} y-1     +\frac{g}{4} \lambda -
\sqrt{2 \alpha+\beta-1} E\left[ \mathrm{am} \left( y \right) \right].
\label{D3}
\ee
Here $E[\mathrm{am} (y)]=E[\mathrm{am}(y),k]$ is the incomplete elliptic integral of the second kind,
\[
E[\mathrm{am}(y),k]= \int_0^{\mathrm{am} (y)} \sqrt{1-k^2 \sin^2 \theta} d \theta =
\int_0^y \mathrm{dn}^2 (w,k) dw,
\]
and $\mathrm{am} (y)$  is the elliptic amplitude.
(To simplify the notation, we omit the dependence on the elliptic modulus $k$ in \eqref{D3}.)

A similar procedure involving the solution $r_B$  yields
\be
\zeta_B(\alpha, \beta, y, \lambda)= \xi,
\label{xib}
\ee
 where we have introduced
\be
\zeta_B= \frac{\alpha+\beta-1}{\sqrt{2 \alpha+\beta-1}} y -1  +\frac{g}{4} \lambda -
\sqrt{2 \alpha+\beta-1}
\left\{ E \left( \frac{\pi}{2} \right) -
 E \left[ \mathrm{am} \left( K- y \right)  \right] \right\}.
 \label{D4}
\ee
Here $E(\pi/2)=E(\pi/2,k)$ is the complete elliptic integral of the second kind.

Note that unlike the pair of equations \eqref{A37} and \eqref{A36},
the normalisation constraints \eqref{xia} and \eqref{xib} are not related by the
transformation $y=K-\widetilde{y}$.
Therefore, the solution of the system \eqref{A36}+\eqref{xib} cannot be reduced to
solving \eqref{A37}+\eqref{xia}. The ``$r_A$" and ``$r_B$"  systems have to be
considered independently.

\section{Reduction to the linear Schr\"odinger equation}
\label{linear_reduction}

Before proceeding to the analysis of the systems  \eqref{A37}+\eqref{xia} and
\eqref{A36}+\eqref{xib},
it is instructive to verify that the transcendental equation  \eqref{C6} for the linear
Gross-Pitaevskii equation
is recovered as the $g \to 0$ limit of (\ref{A36}).

Assume that $\kappa$ is fixed and $g$ is varied in the stationary Gross-Pitaevskii equation \eqref{A1}.
When $g$ is small,
the solution of \eqref{A1} bifurcating from the solution of the corresponding {\it linear\/}  Schr\"odinger equation
 remains of order $1$.
The corresponding solution
 $\varphi$ of equation \eqref{A100}
  will  then have to be of order $g^{1/2}$.
  This means, in particular, that for all $|\tau | \geq T$,
  the ``outer" solution  (\ref{A6}) will have to approach zero as $g \to 0$.
In order to have
\[
\int_T^\infty
\sech^2 (\tau-\mu) d \tau  \to 0   \quad \text{as} \ g \to 0,
\]
 one has to   require that $\xi \to -1$ as $g \to 0$.
Defining $\sigma$ by
\be
\sech^2 (\mu-T) = \sigma^2 g, \quad \sigma=O(1),
\label{W5}
\ee
the quantity $\xi$ will have the following asymptotic behaviour:
\[
\xi = -1+ \frac{\sigma^2}{2} g +O(g^2).
\]
Letting
\[
\alpha= 1+ {\mathcal A}_1 g +   {\mathcal A}_2g^2 + ...., \quad  \beta= {\mathcal B}_1 g+ {\mathcal B}_2 g^2+ ...
\]
and substituting these expansions in  (\ref{A34}) and (\ref{A35}) gives
\[
{\mathcal A}_1= A \sigma^2,
\quad
{\mathcal B}_1 = B \sigma^2,
\]
where
\be
A = \lambda^2 - 2 \lambda+ \eta^2,
\quad
B = -\frac12 A + \frac12 \sqrt{A^2+ 4 \eta^2}.
\label{W1}
\ee

Turning to the transcendental equation (\ref{A36}), we note that the elliptic modulus of the Jacobi cosine tends to 1 as $g \to 0$:
\[
k^2= 1-({\mathcal A}_1+ 2 {\mathcal B}_1) g + ...
\]
In this limit, the elliptic function approaches a hyperbolic sine:
\[
 \mathrm{cn} (K-y,k)= k' \sinh y+ O(k'^3), \quad
 k'^2= 1-k^2.
 \]
In equation (\ref{A36}), $y=\sqrt{2 \alpha +\beta-1} T$.
With $\sqrt{2 \alpha+ \beta-1} = 1+O(g)$,
the transcendental equation
 reduces to
\be
B+ (A+ 2B) \sinh^2 T=1.
\label{W2}
\ee
Substituting for $A$ and $B$ from (\ref{W1}),  equation (\ref{W2}) gives
\be
e^{-4T}= \frac{(\lambda-2)^2+ \eta^2}{\lambda^2+\eta^2}.
\label{W3}
\ee
Transforming to $\gamma ={\eta}/{\lambda}$, $\kappa=1/{\lambda}$  and  $L=2 \lambda T$,
we recover the transcendental equation
 (\ref{C6}).

Finally, we consider
 the normalisation constraint (\ref{xib}).
 As $k \to 1$,
 the elliptic integral of the second kind has the following asymptotic behaviour:
 \[
E\left(\frac{\pi}{2}\right)-E  \left[ \mathrm{am} (K-y) \right] = k'^2 \int_0^y \frac{ du}{{\mathrm{dn}}^2 (u,k)}=
\frac{k'^2}{2} \left( y + \frac{\sinh 2y}{2} \right) +O \left( k'^4 \right).
\]
Making use of this expansion, we reduce Eq.\eqref{xib}
 to
\be
2-2AT+  \sqrt{A^2+4 \eta^2}  \sinh(2T)=\frac{\lambda}{\sigma^2}.
\label{W6}
\ee
Given $\lambda$, $\eta$ and $T$,
  equation (\ref{W6})
furnishes the coefficient $\sigma$
in the relation \eqref{W5}.
The
 relation (\ref{W5}), in turn, determines the amplitude of the
solution $\varphi$ corresponding to the nonlinearity parameter $g$.

\section{Transcendental equations}
\label{Transcendental}

In this section we assume that  $L$, the distance between the potential wells in
the original Gross-Pitaevskii equation  \eqref{A1}, and $g$,
the coefficient of the nonlinearity, are fixed.
On the other hand,  the gain-loss coefficient $\gamma$
and the ``nonlinear eigenvalue"  $\kappa$
(and hence the scaling factor $\lambda=\kappa^{-1}$, the scaled gain-loss $\eta=\gamma \kappa^{-1}$, and
the dimensionless well-separation distance $2T=\kappa L$) are allowed to vary.
The parameter $\xi$ --- the parameter defining the amplitude of the nonlinear mode ---  has not been fixed either.

Substituting $\xi$ from the normalisation constraint (\ref{xia}) to the boundary condition (\ref{A37})
we obtain a transcendental equation
\begin{subequations}
\be
\zeta_A^2+ \beta +
(\alpha-\beta) \mathrm{cn}^2  y -1=0
\label{D8}
\ee
for the parameters of the $r_A$ solution.
In a similar way,
substituting from (\ref{xib}) to (\ref{A36}) we obtain an equation
\be
\zeta_B^2 +  \beta +(\alpha-\beta) \mathrm{cn}^2   \left( K-  y \right) -1=0
  \label{D9}
\ee
\end{subequations}
  for the  solution  $r_B$.
Note that
 the functions $\zeta_A$ and $\zeta_B$,
 defined in \eqref{D3} and \eqref{D4},  depend on $\lambda$ as a parameter.

Substituting $\xi$  from the constraint (\ref{xia})
  to the structural relation (\ref{D200}) gives another transcendental equation for the maximum-centred nonlinear mode:
\begin{subequations}
   \be
 (\lambda +\zeta_A)^2- S_A^2=0.
\label{GA1}
\ee
 Here we are using a new notation
$S_A$  for the combination that was previously denoted
$S$ and given by \eqref{S}.
   In a similar way,
    substituting $\xi$ from (\ref{xib}) to (\ref{D200}) gives an equation for the minimum-centred (the $r_B$) solution:
    \be
(\lambda+\zeta_B)^2 -S_B^2=0,
 \label{D6}
 \ee
 \end{subequations}
 where the same combination  $S$  (defined in  \eqref{S}) has been renamed $S_B$.
 (We are using two different notations for the same combination in order to
 be able to set the variable $S$ to two different values later.)
 Like equations \eqref{D8} and \eqref{D9} before,  equations \eqref{GA1} and \eqref{D6} include
 $\lambda$ as a parameter.

  Eliminating $1-\xi^2$ from the expression \eqref{S}  by means
   of the boundary condition (\ref{A37}),  we specify
   $S_A$:
      \begin{subequations}
   \be
S_A^2=
 \frac{  (\alpha+\beta) (\alpha +\beta-1)-
 \alpha \beta }{      \beta+ (\alpha-\beta) \mathrm{cn}^2 y  }
  - \frac{\alpha \beta (\alpha+\beta    -1)}{\left[     \beta+ (\alpha-\beta) \mathrm{cn}^2 y    \right]^2}
    +1 -
    \beta-  (\alpha-\beta) \mathrm{cn}^2 y.
\label{GA2}
\ee
In order to specify $S_B$, we use  the boundary condition (\ref{A36}) instead:
  \begin{align} S_B^2=
 \frac{  (\alpha+\beta) (\alpha +\beta-1)-
 \alpha \beta }{  \beta+ (\alpha-\beta) \mathrm{cn}^2\left( K-y \right)   }
  - \frac{\alpha \beta (\alpha+\beta  -1)}{\left[    \beta+ (\alpha-\beta) \mathrm{cn}^2\left( K-y  \right)           \right]^2}
  \nonumber    \\
    \phantom{\frac{\frac12}{\sqrt{A^2}}}
  +1
 -\beta-  (\alpha-\beta) \mathrm{cn}^2\left( K-y  \right).
 \label{D06}
 \end{align}
 \end{subequations}

 The system \eqref{D8}, \eqref{GA1}
 with $\zeta_A$ as in \eqref{D3} and
 $S_A$ as in \eqref{GA2},  is a system of two
 equations for two parameters of the solution $r_A$
 (the ``$A$-system"). 
 For the given $L$, $g$ and  $\lambda$, the $A$-system has one or several roots $(\alpha_n, \beta_n)$.
 
 Not all roots define the Gross-Pitaevskii solitons though;  some roots are spurious.
 To filter the spurious roots out, we use the simple rule formulated in section \ref{Mexican}. First,
 we calculate the normalised return time \eqref{RT} and establish whether 
 $2n< T/\Theta< 2n+1$ or $2n-1< T/\Theta< 2n$ for some natural $n$. The former situation corresponds to ${\dot r}(-T)>0$ and the latter to ${\dot r}(-T)<0$.
 Evaluating  the amplitude parameter $\xi$ by means of   \eqref{D3}, we then discard the roots with the sign of $\lambda+\xi$ 
 coincident with the sign of ${\dot r}(-T)$.
 
 Having thus validated
  the genuine roots
for a range of $\lambda$ values, we can use 
 equation (\ref{A37}) to express $1-\xi^2$ through $\alpha(\lambda)$ and  $\beta(\lambda)$, and 
 then employ (\ref{A34}) to obtain $\eta(\lambda)$. Transforming from $\lambda$ and $\eta$ to $\kappa=1/\lambda$ and
 $\gamma=\eta/\lambda$, we arrive at
 \begin{subequations}
\be
\gamma_A(\kappa)=
\frac
{\sqrt{\alpha \beta (\alpha+\beta-1)} \kappa}
{
\beta+ (\alpha-\beta) \mathrm{cn}^2\left( \sqrt{2 \alpha+\beta-1} \kappa L/2\right)
}.
\label{gAB}
\ee

The  transcendental  system
\eqref{D9}, \eqref{D6}, \eqref{D4}, \eqref{D06}
(the ``$B$-system")
is not equivalent to the $A$-system
 and has to be solved independently.
Having determined the roots $\alpha(\lambda), \beta(\lambda)$
and validated them in the same way as  we did with the $A$-roots before, 
we obtain an analogue of the formula  (\ref{gAB}):
\be
\gamma_B(\kappa)=
\frac
{\sqrt{\alpha \beta (\alpha+\beta-1)} \kappa}
{
\beta+ (\alpha-\beta) \mathrm{cn}^2\left( K-\sqrt{2 \alpha+\beta-1}\kappa L/2 \right)
}.
\label{gCD}
\ee
\end{subequations}

The curve $\gamma(\kappa)$ --- or, equivalently, $\kappa(\gamma)$ ---  will constitute the central result of our analysis. 
Each root $(\alpha_n, \beta_n)$ of the $A$- or $B$-system will contribute a branch to this curve.
Before presenting  the $\kappa(\gamma)$ relationships
 for various $L$ and $g$, we note a useful symmetry of the $A$- and $B$-systems.

\section{The dip- and hump-adding transformation}

Consider the $r_A$ solution
and assume ($\alpha$, $\beta$) is a root of the  system \eqref{D8}, \eqref{GA1}, \eqref{D3}, \eqref{GA2}
with parameters $g$, $T$ and $\lambda$. The $A$-system with shifted parameters
\begin{align}
{\tilde T}=T+ 2 \Theta, \quad {\tilde \lambda}=\lambda, \quad
{\tilde g}=g+\Delta g, \nonumber \\
\Delta g=
\frac{8}{\lambda}
\left[ \sqrt{2 \alpha+\beta-1} E- \frac{\alpha+\beta-1}{\sqrt{2 \alpha+ \beta-1}} K \right],
\label{map}
\end{align}
will have the same root  ($\alpha$, $\beta$).
Here $\Theta$ is given by \eqref{hp}, while $K=K(k)$ and $E=E(k)$ are
the complete elliptic integrals of the first and second kind, respectively.

The mapping \eqref{map} adds two units to the normalised return time $T/\Theta$: $T/\Theta \to T/\Theta+2$.
Therefore
\eqref{map} adds two humps and two dips to the $A$-mode with $2n$ dips and $2n \pm 1$ humps.
Note that the expression in the square brackets in \eqref{map} is equal to $\int_0^\Theta r_A^2 d \tau$;
hence $\Delta g>0$ for any $\alpha$ and $\beta$. Therefore the map generates an infinite sequence of
nonlinearity strengths $g, g+\Delta g, g+ 2 \Delta g, ...$ supporting  hump-centred nonlinear modes with an increasing number of lateral crests.

Turning to the $B$-solution and the system \eqref{D9}, \eqref{D6}, \eqref{D4}, \eqref{D06},
the same mapping \eqref{map} transforms this system
into itself. As a result of the application of the mapping \eqref{map},
the $B$-solution with $2n$ humps and $2 n \pm 1$ dips acquires two new humps and two new dips.
The expression in the square brackets in \eqref{map} equals $\int_0^\Theta r_B^2 d \tau$;
hence we have $\Delta g>0$ in the case of the $B$ solution as well.
As for the $A$-solution before, the map \eqref{map} generates an infinite sequence of multicrest dip-centred modes.

\section{\emph{PT}-symmetric localised nonlinear modes}

\begin{figure}[t]
	\center
	\includegraphics[width=8cm]{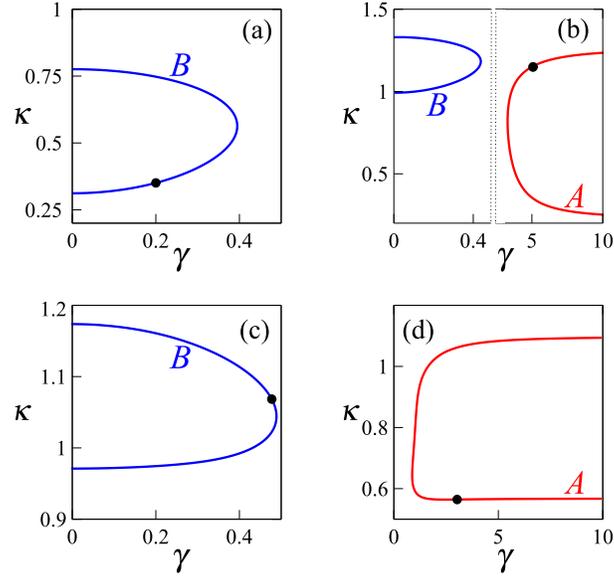}
	%
	%
	\caption{``Nonlinear eigenvalues" $\kappa$ vs the gain-loss coefficient $\gamma$
	for several sets of  $g$ and $L$. 
	The red curves correspond to the $A$- and the  blue ones to the 
	 $B$-modes.  Solutions marked by the 
	 black dots  are shown in Fig.~\ref{fig06}.
	  In these plots, $g=1$, $L=2.2$ (a);  $g=5$, $L=2$ (b);  $g \approx 12.12$, $L \approx 8.26$ (c);  $g \approx 12.38$, $L\approx 9.35$ (d). 
	  Note a break in the horizontal axis in (b). 
	}
	\label{fig05}       
\end{figure}

 The $A$- and  $B$-system 
 of two transcendental equations were solved numerically. 
 We employed a path-following algorithm with a newtonian iteration to obtain the root ($\alpha, \beta$) as $\kappa$ was varied
 with $g$ and $L$  fixed.
 The initial guess for the continuation  process was provided either by the  analysis of intersecting
 graphs of two simultaneous equations on the $(\alpha, \beta)$-plane, or by transplanting a known root 
 to a different set of $g$ and $L$ by means of the mapping \eqref{map}.

  Fig.~\ref{fig05}(a) traces a branch of the $B$-modes on the $(\gamma, \kappa)$-plane. 
  Here, the parameters ($L=2.2$ and $g=1$) correspond to 
   Fig.~1(a) in Ref.\cite{Cartarius1}. 
   These are nonlinear modes  with exactly one dip --- at $x=0$.
   The spatial structure of the mode is illustrated by  Fig. \ref{fig06} (a).

   As it was established numerically in \cite{Cartarius1}
   and corroborated analytically in section \ref{linear_reduction} above,
  the modes making up this branch are nonlinear deformations of the eigenfunctions of the linear Schr\"odinger
   equation (equation \eqref{A1} with $g=0$). 
   This kinship is clearly visible in 
    Fig.~\ref{fig06}(a) where the nonlinear ($g=1$) mode is plotted next to the 
     normalized linear ($g=0$)  eigenfunction with 
       the same value of $\gamma$.

\begin{figure}[t]
	\center
	\includegraphics[width=10cm]{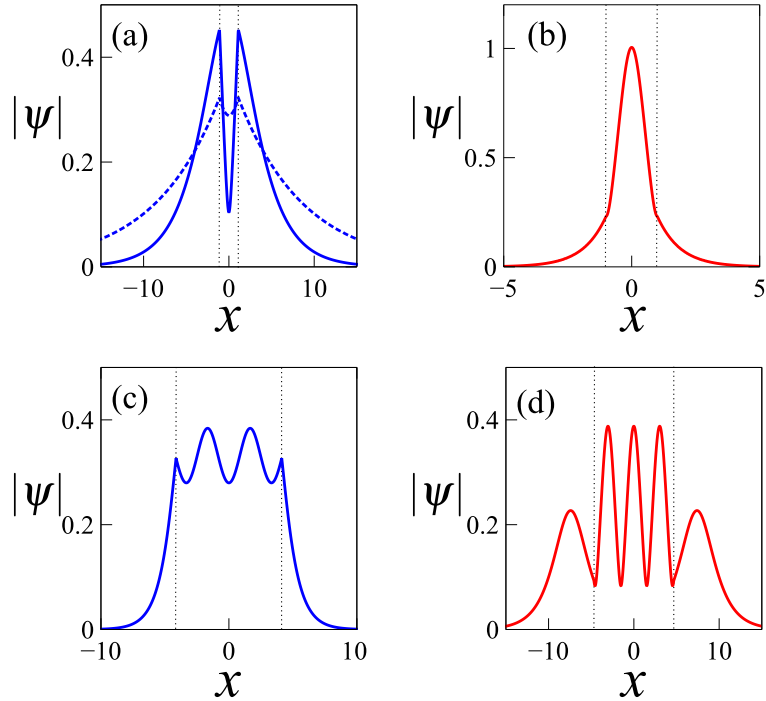}
	%
	%
	\caption{Solid curves depict nonlinear localised modes at representative points along the $\kappa(\gamma)$ curves.
	(These points are marked by black dots in the corresponding panels of Fig.~\ref{fig05}.)
	 The $B$-modes are shown in blue and the $A$-modes in red.   Vertical dotted lines indicate  the positions of the potential wells.
	 The dashed curve in (a) renders the eigenfunction of the equation \eqref{A1} with $g=0$ where $L$ and $\gamma$ are set
	 equal to the $L$ and $\gamma$  of the  nonlinear mode shown in the same panel.
	 Note that the three-hump mode  in (d) has four and not two local minima inside the $(-L,L)$ interval. 
	 The two
	 lateral dips are pressed close to the wells but are nevertheless discernible by zooming in.
	}
	\label{fig06}       
\end{figure}

In contrast to the above $B$ branch, the  $A$-modes exist only 
 when   $g$ exceeds a certain finite threshold; 
 these have no relation  to the $g=0$ eigenfunctions. 
 A single-humped $A$-mode is exemplified by Fig.~\ref{fig06}(b),
 with the corresponding 
 $\kappa(\gamma)$ branch appearing in Fig.~\ref{fig05}(b).

Finally, the bottom panels of Figs.~\ref{fig05}  and \ref{fig06} correspond to nonlinear modes with multiple humps
and dips.
The $\kappa(\gamma)$ curve in 
Fig.~\ref{fig05} (c) pertains to  a 
$B$-mode with three dips and
 two humps  between the potential wells. 
This branch results by the $\kappa$-continuation  from a root ($\alpha_0, \beta_0$) of the $B$-system 
with $\kappa$ equal to some $\kappa_0$ and $g$, $T$ 
  obtained by a once-off application of the map  \eqref{map}.
  A typical nonlinear mode arising along this branch is shown in Fig.~\ref{fig06} (c).

Figure~\ref{fig05}(d) traces a branch of the multi-hump  $A$-modes. 
Solutions on this branch have  three humps  and four dips situated between the wells;
an example is  in Fig.~\ref{fig06}(d).
The starting point for the branch was suggested by the graphical
analysis of the equations making up the $A$-system.

\section{Summary and conclusions}

The  double-$\delta$ well potential, where one well gains and the other one loses particles,
furnishes one of the simplest  
Gross-Pitaevskii 
models employed in the studies of boson condensates.
However the information on its nonlinear modes is scarce and based entirely on numerical solutions.
The purpose of this contribution was to formulate an {\it analytical\/} procedure for the construction
of localised nonlinear modes.

We started with the linear Schr\"odinger equation with the \emph{PT}-symmetric
double-delta well potential
and provided a simple analytical classification of  its bound states.

In the nonlinear situation, 
our procedure reduces the construction of localised modes to finding roots of a system of two 
simultaneous algebraic equations
involving elliptic integrals and Jacobi functions. 
We have classified the nonlinear modes under two broad classes: those with a maximum of $|\psi|$ at the centre
and those centred on a minimum of $|\psi|$.
Accordingly, there are transcendental systems of two types
(referred to as the $A$- and $B$-systems). 
Our construction procedure is supplemented with 
 an ``identification" algorithm allowing to 
relate the number of crests and troughs of the nonlinear mode  to the root of the transcendental system.
We have  established a correspondence
 between localised modes in systems with different 
distances between the wells and different nonlinearity strengths.

Our procedure has been illustrated by the construction of  branches of 
$A$- and $B$-modes for several values of $g$ and $L$.

\begin{acknowledgement}
This contribution is a spin-off from the project on the jamming anomaly in ${\it PT}$-symmetric systems \cite{jam};
we thank Vladimir Konotop for his collaboration on the main part of the project.
Nora Alexeeva's numerical assistance and 
Holger Cartarius' useful remarks are
 gratefully acknowledged. 
 This work was supported by the NRF of South Africa (grants UID 85751, 86991, and 87814) and the FCT (Portugal) through the grants UID/FIS/00618/2013 and PTDC/FIS-OPT/1918/2012.
\end{acknowledgement}
%

%
%

\begin{thebibliography}{99.}%

%
\bibitem{Klaiman} S Klaiman, U G\"unther, and N Moiseyev, 2008.
Visualization of branch points in \emph{PT}-symmetric waveguides. 
Phys Rev Lett {\bf 101}, 080402

\bibitem{Cartarius1}
H Cartarius and G Wunner, 2012.
Model of a  \emph{PT}-symmetric Bose-Einstein condensate in a   $\delta$-function double-well potential.
Phys Rev A {\bf 86} 013612


\bibitem{Musslimani} Z H Musslimani, K G Makris, R El-Ganainy,  and D N Christodoulides, 2008.
Optical solitons in \emph{PT} periodic potentials.
Phys Rev Lett {\bf 100}, 030402

 
   
 
 
 
 
\bibitem{jam} I V Barashenkov, D A Zezyulin and V V Konotop, 2015.
 Jamming anomaly in \emph{PT}-symmetric systems. To be submitted for publication. 
 
 
 
 
\bibitem{GMC}
Yu B Gaididei, S F Mingaleev, and P L Christiansen, 2000.
Curvature-induced symmetry breaking in nonlinear Schr\"odinger models.
Phys Rev E {\bf 62}, R53




\bibitem{JW}
R K Jackson and  M I Weinstein, 2004. Geometric analysis of bifurcation and symmetry breaking in a Gross-Pitaevskii equation.
Journ Stat Phys {\bf 116}, 881



\bibitem{Znojil1} 
M Znojil and V Jakubsk\'y, 2005. Solvability and  \emph{PT}-symmetry in a double-well model
with point interactions. Journ Phys A: Math Gen   {\bf 38}, 5041


\bibitem{Znojil2}
V Jakubsk\'y and
M Znojil, 2005.
An explicitly solvable model of the spontaneous  \emph{PT}-symmetry breaking.
Czechoslovak Journ  Phys    {\bf 55},   1113



\bibitem{Uncu} H Uncu and E Demiralp, 2006.
Bound  state solutions of the Schr\"odinger equation for a \textit{PT}-symmetric
potential with Dirac delta functions.
Phys Lett A {\bf 359}, 190



\bibitem{Cartarius2}
H Cartarius, D Haag, D Dast and G Wunner, 2012.
  Nonlinear Schr\"odinger equation for a   \textit{PT}-symmetric
delta-function double well.
  {J Phys  A: Math Theor} {\bf 45},  444008
  



 









\end{thebibliography}
%

\end{document}